\documentclass[envcountsec,envcountsame]{llncs}

\usepackage{amsfonts, latexsym, amsmath, mathrsfs, amssymb}
\usepackage{algorithm,algpseudocode}
\usepackage{url}
\usepackage{microtype}
\usepackage[colorlinks,citecolor=blue]{hyperref}
\usepackage{fullpage}
\usepackage{times}
\newcommand{\Z}{\mathbb{Z}}

\newcommand{\ve}{\varepsilon}
\newcommand{\randgets}{\overset{\$}{\gets}}

\spnewtheorem{thm}{Theorem}[subsection]{\bfseries}{\itshape}
\spnewtheorem*{defn}{Definition}{\bfseries}{\rmfamily}

\algblockdefx{Repeatx}{EndRepeatx}[1]{\textbf{repeat} #1}%
  {\textbf{end repeat}}
\algloopdefx{Ifx}[1]{\textbf{if} #1 \textbf{then}}

\newif\ifblind\blindfalse
\title{\texorpdfstring{Close to Uniform Prime Number Generation With
Fewer Random Bits\thanks{An extended abstract of this paper will appear
in the proceedings of ICALP~2014. This is the full version. Moreover, an
earlier preprint version of this paper focused on implementation aspects is
also available as IACR ePrint report 2011/481. We intend to merge that
material into this version at a later stage.}}{Close to Uniform Prime Number Generation With 
Fewer Random Bits}}
\ifblind
  \author{(Submission to ICALP 2014)}
  \institute{~}
\else
  \author{Pierre-Alain~Fouque\inst{1} \and Mehdi~Tibouchi\inst{2}}
  \institute{%
     Universit{\'e} de Rennes 1 and Institut universitaire de France \\
     \email{pierre-alain.fouque@ens.fr}
  \and
     NTT Secure Platform Laboratories\\
     \email{tibouchi.mehdi@lab.ntt.co.jp}
  }
\fi

\begin{document}
\maketitle

\begin{abstract}
In this paper, we analyze several variants of a simple method for
generating prime numbers with fewer random bits. To generate a prime $p$
less than $x$, the basic idea is to fix a constant $q\propto x^{1-\ve}$,
pick a uniformly random $a<q$ coprime to $q$, and choose $p$ of the form
$a+t\cdot q$, where only $t$ is updated if the primality test fails. We
prove that variants of this approach provide prime generation algorithms
requiring few random bits and whose output distribution is close to
uniform, under less and less expensive assumptions: first a relatively
strong conjecture by H.~Montgomery, made precise by Friedlander and
Granville; then the Extended Riemann Hypothesis; and finally fully
unconditionally using the Barban--Davenport--Halberstam theorem.

\medskip
We argue that this approach has a number of desirable properties compared
to previous algorithms. In particular:
\begin{itemize}
\item
it uses much fewer random bits than both the ``trivial algorithm''
(testing random numbers less than $x$ for primality) and Maurer's almost
uniform prime generation algorithm;

\item
the distance of its output distribution to uniform can be made
arbitrarily small, unlike algorithms like PRIMEINC (studied by Brandt and
Damg\aa{}rd), which we show exhibit significant biases;

\item
all quality measures (number of primality tests, output entropy,
randomness, etc.) can be obtained under very standard conjectures or even
unconditionally, whereas most previous nontrivial algorithms can only be
proved based on stronger, less standard assumptions like the
Hardy--Littlewood prime tuple conjecture.
\end{itemize}
\end{abstract}

\allowdisplaybreaks

\section{Introduction}
\label{s:intro}

There are several ways in which we could assess the quality of a random
prime generation algorithm, such as its speed (time complexity), its
accuracy (the probability that it outputs numbers that are in fact
composite), its statistical properties (the regularity of the output
distribution), and the number of bits of randomness it consumes to
produce a prime number (as good randomness is crucial to key generation
and not easy to come by \cite{rfc4086}).

In a number of works in the literature, cryptographers have proposed
faster prime
generation algorithms \cite{DBLP:conf/asiacrypt/BrandtDL91,DBLP:conf/crypto/BrandtD92,DBLP:conf/ches/JoyePV00,DBLP:conf/ches/JoyeP06}
or algorithms providing a proof that the generated numbers are indeed prime
numbers \cite{DBLP:conf/eurocrypt/Maurer89,DBLP:journals/joc/Maurer95,DBLP:conf/crypto/Mihailescu94}. 

A number of these works also prove lower bounds on the entropy of the
distribution of prime numbers they generate, usually based on very strong
conjectures on the regularity of prime numbers, such as the prime
$r$-tuple conjecture of Hardy--Littlewood~\cite{HL}. However, such bounds
on the entropy do not ensure that the resulting distribution is
statistically close to the uniform distribution: for example, they do not
preclude the existence of efficient distinguishers from the uniform
distribution, which can indeed be shown to exist in most cases.

But some cryptographic protocols (including most schemes based on the
Strong RSA assumption, such as Cramer--Shoup
signatures~\cite{DBLP:journals/tissec/CramerS00}) specifically require
uniformly distributed prime numbers for the security proofs to go
through.

Moreover, some cryptographers, like
Maurer~\cite{DBLP:conf/eurocrypt/Maurer89}, have argued that even for
more common uses of prime number generation, like RSA key generation, one
should preferably generate primes that are almost uniform, so as to avoid
biases in the RSA moduli $N$ themselves, even if it is not immediately
clear how such biases can help an adversary trying to factor $N$. This
view is counterbalanced by results of Mih\u{a}ilescu~\cite{MR1944733}
stating in particular that, provided the biases are not too large (a
condition that is satisfied by the algorithms with large output entropy
mentioned above, if the conjectures used to establish those entropy
bounds hold), then, asymptotically, they can give at most a
polynomial advantage to an adversary trying to factor $N$. This makes the
problem of uniformity in prime number generation somewhat comparable to
the problem of tightness in security reductions.

To the authors' knowledge, the only known prime generation algorithms for
which the statistical distance to the uniform distribution can be bounded
are the one proposed by
Maurer~\cite{DBLP:conf/eurocrypt/Maurer89,DBLP:journals/joc/Maurer95} on
the one hand, and the trivial algorithm (viz. pick a random odd integer
in the desired interval, return it if it is prime, and try again
otherwise) on the other hand. The output distribution of the trivial
algorithm is exactly uniform (or at least statistically close, once one
accounts for the compositeness probability of the underlying randomized
primality checking algorithm), and the same can be said for at least some
variants of Maurer's algorithm, but both of those algorithms have the
drawback of consuming a very large amount of random bits.

By contrast, the PRIMEINC algorithm studied by Brandt and
Damg{\aa}rd~\cite{DBLP:conf/crypto/BrandtD92} (basically, pick a random
number and increase it until a prime is found) only consumes roughly as
many random bits as the size of the output primes, but we can show that
its output distribution, even if it can be shown to have high entropy if
the prime $r$-tuple conjecture holds, is also provably quite far from
uniform, as we demonstrate in \S\ref{sec:why}. It is likely that
most algorithms that proceed deterministically beyond an initial random
choice, including those of Joye, Paillier and
Vaudenay~\cite{DBLP:conf/ches/JoyePV00,DBLP:conf/ches/JoyeP06},
exhibit similar distributional biases.

The goal of this paper is to achieve in some sense the best of both
worlds: construct a prime generation algorithm that consumes much fewer
random bits than the trivial algorithm while being efficient and having
an output distribution that is provably close to the uniform one.

We present such an algorithm in \S\ref{sec:method}: to generate a prime
$p$, the basic idea is to fix a constant $q\sim x^{1-\ve}$, pick a
uniformly random $a<q$ coprime to $q$, and choose $p$ of the form
$a+t\cdot q$, where only $t$ is updated if the primality test fails. We
prove that variants of this approach provide prime generation algorithms
requiring few random bits and whose output distribution is close to
uniform, under less and less expensive assumptions: first a relatively
strong conjecture by H.\,L.~Montgomery, made precise by Friedlander and
Granville; then the Extended Riemann Hypothesis; and finally fully
unconditionally using the Barban--Davenport--Halberstam theorem.

\section{Preliminaries}
\label{sec:prelims}

\subsection{Regularity measures of finite probability distributions}

In this subsection, we give some definitions on distances between random
variables and the uniform distribution on a finite set. We also provide some
relations which will be useful to bound the entropy of our prime
generation algorithms. These results can be found in~\cite{Shoup}.

%

\begin{defn}[Entropy and Statistical Distance]
Let $X$ and $Y$ be two random variables on a finite set $S$. The
\emph{statistical distance} between them is defined as the $\ell_1$
norm:%
\footnote{An alternate definition frequently found in the literature
differs from this one by a constant factor $1/2$. That constant factor is
irrelevant for our purposes.}
\[ \Delta_1(X;Y)=\sum_{s\in S} \Big| \Pr[X=s] - \Pr[Y=s]\Big|. \]
We simply denote by $\Delta_1(X)$ the statistical distance between $X$
and the uniform distribution on $S$:
\[ \Delta_1(X)=\sum_{s\in S} \Big| \Pr[X=s] - \frac1{|S|}\Big|, \]
and say that $X$ is \emph{statistically close to uniform} when $\Delta_1(X)$ is
negligible.%
\footnote{For this to be well-defined, we of course need a family of
random variables on increasingly large sets $S$. Usual abuses of language
apply.}

The \emph{squared Euclidean imbalance} of $X$ is the square of the
$\ell_2$ norm between $X$ and the uniform distribution on the same set:
\[\Delta^2_2(X)=\sum_{s \in S} \Big| \Pr[X=s]-1/|S|\Big|^2.\]
We also define the \emph{collision probability} of $X$ as: 
\[\beta(X)=\sum_{s\in S} \Pr[X=s]^2,\]
and the \emph{collision entropy} (also known as the R\'enyi entropy)
of $X$ is then $H_2(X) = -\log_2 \beta(X)$. Finally, the
\emph{min-entropy} of $X$ is $H_\infty(X) = -\log_2\gamma(X)$, where
$\gamma(X)=\max_{s\in S}(\Pr[X=s])$.
\end{defn}

\begin{lemma}
Suppose $X$ is a random variable of a finite set $S$.
The quantities defined above satisfy the following relations:
\begin{align}
\gamma(X)^2 &\leq \beta(X) = 1/|S| + \Delta^2_2(X) \leq \gamma(X) \leq
1/|S|+\Delta_1(X), \label{eq:beta}\\
\Delta_1(X) &\leq \Delta_2(X)\sqrt{|S|}. \label{eq:d1d2}
\end{align}
\end{lemma}

\subsection{Prime numbers in arithmetic progressions}


All algorithms proposed in this paper are based on the key idea that, for
any given integer $q>1$, prime numbers are essentially equidistributed
among invertible classes modulo $q$. The first formalization of that idea
is de la Vall{\'e}e Poussin's \emph{prime number theorem for arithmetic
progressions}~\cite{ValleePoussin96}, which states that for any fixed
$q>1$ and any $a$ coprime to $q$, the number $\pi(x;q,a)$ of prime
numbers $p\leq x$ such that $p\equiv a\pmod q$ satisfies:
\begin{equation}
\label{eq:valleepoussin}
\pi(x;q,a) \mathop{\sim}_{x\to+\infty} \frac{\pi(x)}{\varphi(q)}. 
\end{equation}

De la Vall{\'e}e Poussin established that estimate for constant $q$, but it
is believed to hold uniformly in a very large range for $q$. In fact,
H.\,L.~Montgomery conjectured~\cite{Montgomery71,Montgomery76} that for any
$\ve > 0$:%
\footnote{As is usual in analytic number theory and related subjects,
we use the notations $f(u) \ll g(u)$ and $f(u) = O\big(g(u)\big)$
interchangeably. A subscripted variable on $\ll$ or $O$ means that the implied
constant depends only on that variable.}
\[ \left| \pi(x;q,a) - \frac{\pi(x)}{\varphi(q)} \right| \ll_\ve
   (x/q)^{1/2+\ve} \qquad \big( q<x,\ (a,q)=1 \big),
\]
which would imply that \eqref{eq:valleepoussin} holds uniformly for $q
\ll x/\log^{2+\ve} x$. However, Friedlander and Granville showed~\cite{FG89}
that conjecture to be overly optimistic, and proposed the following
corrected estimate.
\begin{conjecture}[Friedlander--Granville--Montgomery]
\label{conj:fgm}
For $q<x$, $(a,q) = 1$ and all $\ve > 0$, we have:
\[ \left| \pi(x;q,a) - \frac{\pi(x)}{\varphi(q)} \right| \ll_\ve
   (x/q)^{1/2}\cdot x^\ve.
\]
In particular, the estimate \eqref{eq:valleepoussin} holds uniformly for
$q \ll x^{1-3\ve}$.
\end{conjecture}

That conjecture is much more precise than what can be proved using
current techniques, however. The best unconditional result of the same
form is the Siegel--Walfisz theorem~\cite{Walfisz36}, which only implies
that \eqref{eq:valleepoussin} holds in the much smaller range $q \ll
(\log x)^A$ (for any $A>0$).

Stronger estimates can be established assuming the Extended Riemann
Hypothesis ({i.e.} the Riemann Hypothesis for $L$-functions of Dirichlet
characters), which
gives~\cite[p.~125]{Davenport}:
\[ \left| \pi(x;q,a) - \frac{\pi(x)}{\varphi(q)} \right| \ll
   x^{1/2}\log x \qquad \big( q<x,\ (a,q)=1 \big).
\]
This implies \eqref{eq:valleepoussin} in the range $q\ll
x^{1/2}/\log^{2+\ve} x$, which is again much smaller than the one from
Conjecture~\ref{conj:fgm}. The range can be extended using averaging,
however. The previous result under ERH is actually deduced from estimates
on the character sums $\pi(x,\chi) = \sum_{p\leq x} \chi(p)$ for
nontrivial Dirichlet characters $\chi \bmod q$, and more careful
character sum arguments allowed Tur{\'a}n to obtain the following
theorem.
\begin{theorem}[Tur{\'a}n~\cite{Turan37}]
The Extended Riemann Hypothesis implies that for all $q<x$:
\[ \sum_{a\in(\Z/q\Z)^*} \left| \pi(x;q,a) - \frac{\pi(x)}{\varphi(q)}
\right|^2 \ll x(\log x)^2
\]
where the implied constant is absolute.
\end{theorem}
That estimate is nontrivial in the large range $q\ll x/\log^{4+\ve}$,
and implies that \eqref{eq:valleepoussin} holds for all $q$ in that range
and \emph{almost all} $a\in(\Z/q\Z)^*$.

Averaging over the modulus as well, it is possible to obtain fully
unconditional estimates valid in a similarly wide range: this is a result
due to Barban~\cite{Barban66} and Davenport and Halberstam~\cite{DH66}.
We will use the following formulation due to
Gallagher~\cite{Gallagher67}, as stated in \cite[Ch.~29]{Davenport}.
\begin{theorem}[Barban--Davenport--Halberstam]
For any fixed $A > 0$ and any $Q$ such that $x(\log x)^{-A} < Q < x$, we
have:
\[ \sum_{q\leq Q} \sum_{a\in(\Z/q\Z)^*}
   \left| \pi(x;q,a) - \frac{\pi(x)}{\varphi(q)} \right|^2 \ll_A
   \frac{xQ}{\log x}.
\]
\end{theorem}

Finally, we will also need a few classical facts regarding Euler's totient
function (for example, \cite[Th. 328 \& 330]{hardy}).
\begin{lemma}
The following asymptotic estimates hold:
\begin{align}
\varphi(q) &\gg \frac{q}{\log\log q}, \\
\Phi(x)    &:=  \sum_{q\leq x} \varphi(q) = \frac{3x^2}{\pi^2} + O(x\log x).
\end{align}
\end{lemma}

\section{Close-to-uniform prime number generation with fewer random bits}
\label{sec:method}

\subsection{Basic algorithm}

A simple method to construct obviously uniformly distributed prime
numbers up to $x$ is to pick random numbers in $\{1,\dots,\lfloor
x\rfloor\}$ and retry until a prime is found. However, this method
consumes $\log_2 x$ bits of randomness per iteration (not counting the
amount of randomness consumed by primality testing), and hence an
expected amount of $(\log x)^2/\log 2$ bits of randomness to produce a
prime, which is quite large.

As mentioned in the introduction, we propose the following algorithm to
generate almost uniform primes while consuming fewer random bits: first
fix an integer
\begin{equation}
\label{eq:rangeq}
q\propto x^{1-\ve}
\end{equation}
and pick a random $a\in(\Z/q\Z)^*$. Then, search for prime numbers
$\leq x$ of the form $p = a + t\cdot q$.
This method, described as Algorithm~\ref{alg:basic}, only consumes
$\log_2 t = \ve\log_2 x$ bits of randomness per iteration, and the
probability of success at each iteration is 
$\sim \frac{\pi(x;q,a)}{x/q}$. Assuming that Conjecture~\ref{conj:fgm} is
true, which ensure that~\eqref{eq:valleepoussin} holds in the
range~\eqref{eq:rangeq}, this probability is about $q/\big(\varphi(q)\log
x\big)$, and the algorithm should thus consume roughly:
\begin{equation}
\label{eq:basicrandbits}
\ve\cdot\frac{\varphi(q)}{q}\cdot\frac{(\log x)^2}{\log 2}
\end{equation}
bits of
randomness on average: much less than the trivial algorihm. Moreover, we
can also show, under the same assumption, that the output distribution
is statistically close to uniform and has close to maximal entropy.

\begin{algorithm}[t]
\caption{Our basic algorithm.}
\label{alg:basic}
\begin{algorithmic}[1]
  \State Fix $q \propto x^{1-\ve}$
  \State $a \randgets (\Z/q\Z)^*$
    \label{algstep:basic-a}
    \Comment{considered as an element of $\{1,\dots,q-1\}$}
  \Repeatx{forever}
    \State $t \randgets \{0,\dots,\lfloor \frac{x-a}q \rfloor\}$
    \State $p \gets a + t\cdot q$
    \Ifx{$p$ is prime}
      \label{algstep:test}
      \textbf{return} $p$
  \EndRepeatx
\end{algorithmic}
\end{algorithm}

We establish those results in \S\ref{sec:fgm}, and show in
\S\ref{sec:basic-erh} that Tur{\'a}n's theorem can be used to obtain
nearly the same results under the Extended Riemann Hypothesis. ERH is not
sufficient to prove that Algorithm~\ref{alg:basic} terminates almost
surely, or to bound the expectation of the number of random bits it
consumes, due to the possibly large contribution of negligibly few values
of $a$. We can avoid these problems by modifying the algorithm slightly,
as discussed in \S\ref{sec:as-erh}. Finally, in \S\ref{sec:bdh}, we show
that unconditional results of the same type can be obtained using the
Barban--Davenport--Halberstam theorem, for another slightly different
variant of the algorithm.

Before turning to these analyses, let us make a couple of remarks on
Algorithm~\ref{alg:basic}. First, note that one is free to choose $q$ in
any convenient way in the range~\eqref{eq:rangeq}. For example, one could
choose $q$ as the largest power of $2$ less than $x^{1-\ve}$, so as to
make Step~\ref{algstep:basic-a} very easy. It is preferable, however, to
choose $q$ as a (small multiple of a) primorial, to minimize the ratio
$\varphi(q)/q$, making it as small as $\propto 1/\log\log q \sim
1/\log\log x$; this makes the expected number of iterations and the
expected amount~\eqref{eq:basicrandbits} of consumed randomness
substantially smaller. In that case, Step~\ref{algstep:basic-a} becomes
slightly more complicated, but this is of no consequence.

Indeed, our second observation is that Step~\ref{algstep:basic-a} is
always negligible in terms of running time and consumed randomness
compared to the primality testing loop that follows. Indeed, even the
trivial implementation (namely, pick a random $a\in\{0,\dots,q-1\}$ and
try again if $\gcd(a,q)\neq 1$) requires $q/\varphi(q) \ll \log\log q$
iterations on average. It is thus obviously much faster than the
primality testing loop, and consumes $\ll \log x\log\log x$ bits of
randomness, which is negligible compared to~\eqref{eq:basicrandbits}.
Furthermore, an actual implementation would take advantage of the known
factorization of $q$ and use a unit generation algorithm such as the one
proposed by Joye and Paillier~\cite{DBLP:conf/ches/JoyeP06}, which we can
show requires only $O(1)$ iterations on average.

Finally, while we will not discuss the details of the primality test of
Step~\ref{algstep:test}, and shall pretend that it returns exact results
(as the AKS algorithm~\cite{AKS} would, for example), we note that it is
fine (and in practice preferable) to use a probabilistic compositeness
test such as Miller--Rabin~\cite{Rabin} instead, provided that the number
of rounds is set sufficiently large as to make the error probability
negligible. Indeed, the output distribution of our algorithm then stays
statistically close to uniform, and the number of iterations is never
larger.

\subsection{Analysis under the Friedlander--Granville--Montgomery
conjecture}
\label{sec:fgm}

As mentioned above, it is straightforward to deduce from the
Friedlander--Granville--Montgomery conjecture that
Algorithm~\ref{alg:basic} terminates almost surely, and to bound its
expected number of iterations and amount of consumed randomness.

\begin{thm}
\label{th:termfgm}
Assume that Conjecture~\ref{conj:fgm} holds. Then
Algorithm~\ref{alg:basic} terminates almost surely, requires
$(1+o(1))\varphi(q)/q\cdot\log x$ iterations of the main loop on average,
and consumes
$ \big(\ve + o(1)\big)\cdot\frac{\varphi(q)}{q}\cdot\frac{(\log x)^2}{\log 2}
$
bits of randomness on average.
\end{thm}
\begin{proof}
Indeed, fix $q\propto x^{1-\ve}$.
Conjecture~\ref{conj:fgm} implies, uniformly over $a\in(\Z/q\Z)^*$:
\[ \left| \pi(x;q,a) - \frac{\pi(x)}{\varphi(q)} \right| \ll
   (x/q)^{1/2}\cdot x^{\ve/4} \propto x^{3\ve/4},
\]
which is negligible compared to $\pi(x)/\varphi(q) \gg x^\ve/\log x$. As
a result, we get $\pi(x;q,a) = (1+o(1))\pi(x)/\varphi(q) =
(1+o(1))/\varphi(q)\cdot x/\log x$ uniformly over
$a$, and the success probability of the main loop becomes:
\[ \frac{\pi(x;q,a)}{1 + \lfloor\frac{x-a}q\rfloor} =
\frac{q}{\varphi(q)}\cdot\frac{1+o(1)}{\log x} \]
which implies the stated results immediately.\qed
\end{proof}

Now let $X$ be the output distribution of Algorithm~\ref{alg:basic},
{i.e.} the distribution on the set of prime numbers $\leq x$ such that
Algorithm~\ref{alg:basic} outputs a prime $p$ with probability exactly
$\Pr[X = p]$. Clearly, we have, for all $(a,q)=1$ and all $t$ such that
$a+t\cdot q\leq x$ is prime:
\[ \Pr[X = a + t\cdot q] = \frac1{\varphi(q)}\cdot\frac1{\pi(x;q,a)}. \]
As a result, the squared Euclidean imbalance of $X$ is:
\begin{align*}
   \Delta_2^2(X)
&= \sum_{a\in(\Z/q\Z)^*} \sum_{a+tq\leq x\text{ prime}}
      \Big|\Pr[X=a+tq] - \frac1{\pi(x)}\Big|^2 +
      \sum_{p|q}\frac1{\pi(x)^2} \\
&= \sum_{a\in(\Z/q\Z)^*} \pi(x;q,a)
      \Big|\frac1{\varphi(q)}\cdot\frac1{\pi(x;q,a)} - \frac1{\pi(x)}\Big|^2 +
      \sum_{p|q}\frac1{\pi(x)^2} \\
&= \frac1{\pi(x)^2} \sum_{a\in(\Z/q\Z)^*} \frac1{\pi(x;q,a)}
      \Big|\pi(x;q,a) - \frac{\pi(x)}{\varphi(q)}\Big|^2 +
      \sum_{p|q}\frac1{\pi(x)^2} \\
&\ll \frac1{\pi(x)^2} \sum_{a\in(\Z/q\Z)^*} \frac{\log x}{x^\ve}\cdot x^{3\ve/2}
 \ll \frac{\log^3 x}{x^2} \cdot \varphi(q)x^{\ve/2}
 \ll \frac{\log^3 x}{x^{1+\ve/2}} \ll \frac1{x^{1+\ve/3}}.
\end{align*}
We can then deduce the following.

\begin{thm}
\label{th:statfgm}
Assume that Conjecture~\ref{conj:fgm} holds. Then the output distribution
of Algorithm~\ref{alg:basic} is statistically close to uniform, and its
collision entropy is only negligibly smaller than that of the uniform
distribution.
\end{thm}
\begin{proof}
Indeed, by~\eqref{eq:d1d2}, the statistical distance to the uniform 
distribution satisfies:
\[ \Delta_1(X) \leq \Delta_2(X)\sqrt{\pi(x)}
     \ll \frac{1}{x^{1/2+\ve/6}} \sqrt{\frac{x}{\log x}} \ll x^{-\ve/6},
\]
which is negligible. Moreover, the collision probability is:
\[ \beta(X) = \frac1{\pi(x)} + \Delta_2^2(X)
     = \frac1{\pi(x)}\left(1 + O\Big(\frac{\pi(x)}{x^{1+\ve/3}}\Big)\right)
     = \frac1{\pi(x)}\left(1 + o\big(x^{-\ve/3}\big)\right).
 \]
Hence:
\[ H_2(X) = \log_2\big(\pi(x)\big) - \log_2\big(1 + o(x^{-\ve/3})\big)
          = (H_2)_{\max} - o(x^{-\ve/3})
\]
as required.\qed
\end{proof}

\subsection{Analysis under the Extended Riemann Hypothesis}
\label{sec:basic-erh}

Assume the Extended Riemann Hypothesis, and denote by $\alpha$ the
fraction of all possible choices of $a\in(\Z/q\Z)^*$ such that the error
term $E(x;q,a) := \big| \pi(x;q,a) - \pi(x)/\varphi(q) \big|$ satisfies
$E(x;q,a) > x^{3\ve/4}$. Then, Tur{\'a}n's theorem asserts that:
\[ \sum_{a\in(\Z/q\Z)^*} E(x;q,a)^2 \ll x(\log x)^2, \]
and the left-hand side is greater or equal to
$\alpha \varphi(q)\cdot x^{3\ve/2}$ by
definition of $\alpha$. As a result, we get:
\[ \alpha \ll \frac{x^{1-3\ve/2}(\log x)^2}{\varphi(q)}
          \ll \frac{(\log x)^2\log\log x}{x^{\ve/2}}
\]
and hence $\alpha$ is negligible. Therefore, for all except at most a
negligible fraction of choices of $a\in(\Z/q\Z)^*$, we obtain that
$E(x;q,a) \leq x^{3\ve/4}$, and since $\pi(x)/\varphi(q) \gg x^\ve/\log
x$, this implies $\pi(x;q,a) = (1+o(1))\pi(x)/\varphi(q)$ as before. As a
result, under ERH, we obtain an analogue of Theorem~\ref{th:termfgm}
valid with overwhelming probability on the choice of $a$.
\begin{thm}
\label{th:termbasic-erh}
Assume ERH holds. Then
Algorithm~\ref{alg:basic} terminates with overwhelming probability.
Moreover, except for a negligible fraction of choices of the
class $a \bmod q$, it requires
$(1+o(1))\varphi(q)/q\cdot\log x$ iterations of the main loop on average,
and consumes
$ \big(\ve + o(1)\big)\cdot\frac{\varphi(q)}{q}\cdot\frac{(\log x)^2}{\log 2}
$
bits of randomness on average.
\end{thm}
Moreover, using Tur{\'a}n's theorem and the Cauchy--Schwarz inequality,
we can also establish under ERH alone the following analogue of
Theorem~\ref{th:statfgm}, regarding the output distribution of the
algorithm.
\begin{thm}
\label{th:stat1erh}
Assume ERH holds. Then the output distribution
of Algorithm~\ref{alg:basic} is statistically close to uniform, and its
collision entropy is no more than $O(\log\log x)$ bits smaller than that of the
uniform distribution.
\end{thm}
\begin{proof}
Algorithm~\ref{alg:basic} almost surely produces an output for a given
choice of $a$ if and only if $\pi(x;q,a)\neq0$, and this is no longer
certain under ERH. Therefore, the probability that the algorithm outputs
a prime $p = a + tq\leq x$ becomes:
\[ \Pr[X = a + tq] = \frac1{\varphi^*_x(q)}\cdot\frac1{\pi(x;q,a)}, \]
where $\varphi^*_x(q) = \#\{ a\in(\Z/q\Z)^*\ |\ \pi(x;q,a)\neq0 \}$. By
the previous discussion on the distribution of the values $\pi(x;q,a)$,
we know that $\varphi^*_x(q) = \varphi(q)\cdot\big(1-O(\alpha)\big)$.
As a result, a similar computation as in \S\ref{sec:fgm} gives:
\[ \Delta_2^2(X) = \frac1{\pi(x)^2}
     \sum_{\substack{a\in(\Z/q\Z)^* \\ \pi(x;q,a)\neq0}}
     \frac1{\pi(x;q,a)}\Big|\pi(x;q,a) - \frac{\pi(x)}{\varphi^*_x(q)}\Big|^2
   + \frac{\omega(q)}{\pi(x)^2},
\]
where $\omega(q)$ denotes as usual the number of prime factors of $q$.
Then, using the coarse lower bound $\pi(x;q,a)\geq1$ when
$\pi(x;q,a)\neq0$, we get:
\[ \Delta_2^2(X) \leq \frac{S^2 + \omega(q)}{\pi(x)^2} \]
where:
\begin{align*}
S  &= \sqrt{\sum_{a\in(\Z/q\Z)^*} 
        \Big|\pi(x;q,a) - \frac{\pi(x)}{\varphi^*_x(q)}\Big|^2} \\
&\leq \sqrt{\sum_{a\in(\Z/q\Z)^*} 
        \Big|\pi(x;q,a) - \frac{\pi(x)}{\varphi(q)}\Big|^2} +
      \sqrt{\sum_{a\in(\Z/q\Z)^*} \pi(x)^2 
        \Big|\frac{\varphi(q) - \varphi^*_x(q)}{\varphi(q)\cdot
        \varphi^*_x(q)}\Big|^2} \\
&=    O(x^{1/2}\log x) + \sqrt{\varphi(q)\pi(x)^2
        \Big|\frac{\alpha}{1-\alpha}\cdot\frac1{\varphi(q)}\Big|^2} \\
&\ll  x^{1/2}\log x + \frac{\pi(x)}{\varphi(q)^{1/2}}\cdot\alpha
 \ll  x^{1/2}\log x + \frac{x\log x(\log\log x)^2}{x^{1/2-\ve/2}x^{\ve/2}}
 \ll  x^{1/2}\log x(\log\log x)^2
\end{align*}
by Tur{\'a}n's theorem again. Hence:
\[ \Delta_2^2(X) \ll \frac{x\log^{2+\ve} x}{\pi(x)^2} \ll
    \frac{\log^{3+\ve}
    x}{\pi(x)}.
\]
This is enough to obtain a relatively good bound on the collision entropy:
\[ H_2(X) = \log_2(\pi(x)) - \log_2(\log^{3+\ve} x) = (H_2)_{\max} - O(\log\log x) \]
but isn't sufficient for bounding the statistical distance. However, a
direct computation using Tur{\'a}n's theorem and the Cauchy--Schwarz
inequality is enough:
\begin{align*}
   \Delta_1(X)
&= \sum_{\substack{a\in(\Z/q\Z)^* \\ \pi(x;q,a)\neq0}} \pi(x;q,a)
      \Big|\frac1{\varphi^*_x(q)}\cdot\frac1{\pi(x;q,a)} - \frac1{\pi(x)}\Big| +
   \sum_{p|q} \frac1{\pi(x)} \\
&= \frac1{\pi(x)} \sum_{\substack{a\in(\Z/q\Z)^* \\ \pi(x;q,a)\neq0}} 
      \Big|\pi(x;q,a) - \frac{\pi(x)}{\varphi^*_x(q)}\Big| +
   \frac{\omega(q)}{\pi(x)} \\
&\ll \frac1{\pi(x)}\cdot S\cdot\sqrt{\varphi^*_x(q)}
 \ll \frac{\log x}x\cdot x^{1/2}\log^2 x\cdot\sqrt{q}
 \ll \frac{\log^3 x}{x^{1/2}} \cdot x^{1/2-\ve/2}
 \ll \frac{\log^3 x}{x^{\ve/2}},
\end{align*}
which proves that the distribution is indeed statistically close to
uniform.\qed
\end{proof}

\subsection{Achieving almost sure termination under ERH}
\label{sec:as-erh}

Theorem~\ref{th:termbasic-erh} above is somewhat unsatisfactory, as we
have to ignore a negligible but possibly nonzero fraction of all values
$a\bmod q$ to obtain a bound on the average number of iterations and on
the randomness consumed by Algorithm~\ref{alg:basic} under ERH. But this
is unavoidable for that algorithm: as mentioned above, it is not known
whether ERH implies that for $q\propto x^{1-\ve}$, all $a\in(\Z/q\Z)^*$
satisfy $\pi(x;q,a)\neq 0$. And if an $a$ exists such that
$\pi(x;q,a)=0$, the choice of that $a$ in Step~\ref{algstep:basic-a} of
Algorithm~\ref{alg:basic}, however unlikely, is a case of
non-termination: as a result, the existence of such an $a$ prevents any
nontrivial bound on average running time or average randomness.

We propose to circumvent that problem by falling back to the trivial
algorithm (pick a random $p<x$, check whether it is prime and try again
if not) in case too many iterations of the main loop have been carried
out. This variant is presented as Algorithm~\ref{alg:erh}.

\begin{algorithm}[b]
\caption{A variant which terminates almost surely under ERH.}
\label{alg:erh}
\begin{algorithmic}[1]
  \State Fix $q \propto x^{1-\ve}$
  \State $a \randgets (\Z/q\Z)^*$
    \label{algstep:erh-a}
    \Comment{considered as an element of $\{1,\dots,q-1\}$}
  \Repeatx{$T=\log^2 x$ times}
    \State $t \randgets \{0,\dots,\lfloor \frac{x-a}q \rfloor\}$
    \State $p \gets a + t\cdot q$
    \Ifx{$p$ is prime}
      \textbf{return} $p$
  \EndRepeatx
  \Repeatx{forever}
    \State $p \randgets \{1,\dots,\lfloor x \rfloor\}$
    \Ifx{$p$ is prime}
      \textbf{return} $p$
  \EndRepeatx
\end{algorithmic}
\end{algorithm}

Clearly, since Algorithm~\ref{alg:erh} is the same as
Algorithm~\ref{alg:basic} except for the possible fallback to the trivial
algorithm, which has a perfectly uniform output distribution, the output
distribution of the variant is at least as close to uniform as the
original algorithm. In other words, the analogue of
Theorem~\ref{th:stat1erh} holds, with the same proof.
\begin{thm}
\label{th:stat2erh}
Assume ERH holds. Then the output distribution
of Algorithm~\ref{alg:erh} is statistically close to uniform, and its
collision entropy is no more than $O(\log\log x)$ bits smaller than that of the
uniform distribution.
\end{thm}
Moreover, as claimed above, we can obtain the following stronger analogue
of Theorem~\ref{th:termbasic-erh}.
\begin{thm}
\label{th:term2erh}
Assume ERH holds. Then
Algorithm~\ref{alg:erh} terminates almost surely, requires
$(1+o(1))\varphi(q)/q\cdot\log x$ iterations of the main loop on average,
and consumes
$ \big(\ve + o(1)\big)\cdot\frac{\varphi(q)}{q}\cdot\frac{(\log x)^2}{\log 2}
$
bits of randomness on average.
\end{thm}
\begin{proof}
Algorithm~\ref{alg:erh} terminates almost surely because the trivial
algorithm does. One can estimate its average number of iterations as
follows. Denote by $\varpi(t)$ the probability that
Algorithm~\ref{alg:erh} terminates after exactly $t$ iterations, and
$\varpi_a(t)$ the probability of the same event conditionally to $a$
being chosen in Step~\ref{algstep:erh-a}. We have:
\begin{align*}
  \varpi_a(t) &= \begin{cases}
     \displaystyle
     \bigg(1 - \frac{\pi(x;q,a)}{1+\lfloor\frac{x-a}q\rfloor}\bigg)^{t-1}
     \cdot\frac{\pi(x;q,a)}{1+\lfloor\frac{x-a}q\rfloor}
     & \text{for $t \leq T$;} \\
     \displaystyle
     \bigg(1 - \frac{\pi(x;q,a)}{1+\lfloor\frac{x-a}q\rfloor}\bigg)^{T}
     \bigg(1 - \frac{\pi(x)}{\lfloor x\rfloor}\bigg)^{t-T-1}
     \cdot\frac{\pi(x)}{\lfloor x\rfloor}
     & \text{otherwise.}
   \end{cases} \\
  \varpi(t) &= \frac1{\varphi(q)}\sum_{a\in(\Z/q\Z)^*} \varpi_a(t).
\end{align*} 
Moreover, the expected number $N$ of iterations in
Algorithm~\ref{alg:erh} is given by $N = \sum_{t\geq 1} t\varpi(t)$.
We can denote by $N_a = \sum_{t\geq 1} t\varpi_a(t)$ the contribution of
a certain choice $a\in(\Z/q\Z)^*$.

Now, recall from \S\ref{sec:basic-erh} that $\pi(x;q,a)$ is within a
distance $\ll x^{3\ve/4}\log x$ of $\pi(x)/\varphi(q)$, except
for a fraction $\alpha\ll (\log x)^3/x^{\ve/2}$ of all possible
choices of $a$. If we
denote by $A$ the set of ``bad'' choices of $a$, we can write, for all
$a\in A$:
\[
  \varpi_a(t) \leq \begin{cases}
     \displaystyle
     1
     & \text{for $t \leq T$;} \\
     \displaystyle
     \Big(1 - \frac{\pi(x)}{\lfloor x\rfloor}\Big)^{t-T-1}
     \cdot\frac{\pi(x)}{\lfloor x\rfloor}
     & \text{otherwise.}
   \end{cases} \\
\]
Hence, if we let $\xi := \pi(x)/\lfloor x\rfloor$, we get:
\begin{align*}
 N_a &\leq \sum_{t=1}^T t + \sum_{t=T+1}^{+\infty} t(1-\xi)^{t-T-1} \xi
  = T(T+1)/2 + \sum_{k=1}^{+\infty} (T+k)(1-\xi)^{k-1}\xi \\
 N_a &\leq T(T+1)/2 + T\frac{\xi}{\xi}
+ \frac{\xi}{\xi^2} \leq T(T+3)/2 + \frac1\xi \ll \log^4 x.
\end{align*}
On the other hand, for $a\not\in A$, we have $\xi_a :=
\frac{\pi(x;q,a)}{1+\lfloor\frac{x-a}q\rfloor} =
\frac{q}{\varphi(q)}\cdot\frac{1+o(1)}{\log x}$. Therefore:
\begin{align*}
  N_a &= \sum_{t=1}^T t(1-\xi_a)^{t-1}\xi_a + (1-\xi_a)^T \sum_{t=T+1}^{+\infty}
t(1-\xi)^{t-T-1} \xi \\
  N_a &= \frac1{\xi_a} - \sum_{t=T+1}^{+\infty} t(1-\xi_a)^{t-1}\xi_a +
         (1-\xi_a)^T \sum_{t=T+1}^{+\infty} t(1-\xi)^{t-T-1} \xi \\
  \Big| N_a - \frac1{\xi_a} \Big| &\leq (1-\xi_a)^T
	 \sum_{k=1}^{+\infty} \Big[ (T+k)(1-\xi)^{k-1} \xi 
	 + (T+k)(1-\xi_a)^{k-1} \xi_a \Big] \\
  \Big| N_a - \frac1{\xi_a} \Big| &\leq
	 \exp(-T\xi_a)\cdot(2T+1/\xi+1/\xi_a) \\
  \Big| N_a - \frac1{\xi_a} \Big| &\leq
	 \exp\Big(-(1+o(1))\frac{q}{\varphi(q)}\log x\Big)
	 \cdot(2T+1/\xi+1/\xi_a) \ll \frac1{x^{1-\ve}}.
\end{align*}
As a result, we obtain:
\begin{align*}
  N &= \frac1{\varphi(q)}\sum_{a\in(\Z/q\Z)^*} N_a
     = \Big(1-O(\log^3 x/x^{\ve/2})\Big)\cdot\Big(\frac1{\xi_a} +
       O(1/x^{1-\ve})\Big) + O(\log^3 x/x^{\ve/2})\cdot O(\log^4 x) \\
    &= \frac1{\xi_a} + O\Big(\frac{\log^7 x}{x^{\ve/2}}\Big)
     = (1+o(1))\frac{\varphi(q)}q\cdot \log x
 \end{align*}
as required. As for the expected number $R$ of random bits consumed by
the algorithm, it is given (ignoring the negligible amount necessary to
pick $a$) by:
\[ R = \frac{\log x}{\log 2}\bigg( \sum_{t=1}^T \ve t\cdot\varpi(t) +
\sum_{t=T+1}^{+\infty} (\ve T + t - T)\cdot\varpi(t) \bigg) \]
and the stated estimate is obtained by an exactly analogous computation.
\qed
\end{proof}

\subsection{An unconditional algorithm}
\label{sec:bdh}

Finally, we propose yet another variant of our algorithm for which both
almost sure termination and uniformity bounds can be established
unconditionally. 
The idea is to no longer use a fixed modulus $q$, but to pick it
uniformly at random instead in the range $\{1,\dots,Q\}$ where $Q\propto
x(\log x)^{-A}$; uniformity bounds can then be deduced from the
Barban--Davenport--Halberstam theorem. Unfortunately, since $Q$ is only
polynomially smaller than $x$, we can no longer prove that the output
distribution is statistically close to uniform: the statistical distance
is polynomially small instead, with an arbitrarily large exponent
depending only on the constant $A$. On the other hand, termination is
obtained as before by falling back to the trivial algorithm after a
while, and since $q$ is often very close to $x$, we get an even better
bound on the number of consumed random bits.

\begin{algorithm}[t]
\caption{An unconditional variant.}
\label{alg:uncond}
\begin{algorithmic}[1]
  \State Fix $Q \propto x(\log x)^{-A}$ even
  \State $q \randgets \{Q/2+1,\dots,Q\}$
    \label{algstep:uncond_aq}
  \State $a \randgets \{0,\dots,q-1\}$
  \State \textbf{if} {$\gcd(a,q)\neq1$}
    \textbf{then goto} step \ref{algstep:uncond_aq}
    \label{algstep:uncond_aqcheck}
  \Repeatx{$T=\log^2 x$ times}
    \State $t \randgets \{0,\dots,\lfloor \frac{x-a}q \rfloor\}$
    \State $p \gets a + t\cdot q$
    \Ifx{$p$ is prime}
      \textbf{return} $p$
  \EndRepeatx
  \Repeatx{forever}
    \State $p \randgets \{1,\dots,\lfloor x \rfloor\}$
    \Ifx{$p$ is prime}
      \textbf{return} $p$
  \EndRepeatx
\end{algorithmic}
\end{algorithm}

Our proposed unconditional algorithm is described as
Algorithm~\ref{alg:uncond}. It picks the pair $(q,a)$ uniformly at random
among pairs of integers such that $q\in\{Q/2+1,\dots, Q\}$ and $a$ is a
standard representative of the classes in $(\Z/q\Z)^*$. There are:
\[ F(Q) := \sum_{Q/2<q\leq Q}\varphi(q) = \Phi(Q)-\Phi(Q/2) = \frac9{4\pi^2}Q^2 + O(Q\log Q) \]
possible such pairs, and we claim that for all except a polynomially
small fraction of them, $\pi(x;q,a)$ is close to $\pi(x)/\varphi(q)$.
Indeed, denote by $\alpha$ the fraction of all pairs $(q,a)$ such
that:
\[ E(x;q,a) := \Big|\pi(x;q,a) - \frac{\pi(x)}{\varphi(q)}\Big| >
   (\log x)^{3A/4}.
\] 
Since $\pi(x)/\varphi(q) \gg x/Q \propto (\log x)^A$, we get
$\pi(x;q,a) = (1+o(1))\pi(x)/\varphi(q)$ for all pairs $(q,a)$ except a
fraction of at most $\alpha$. Moreover, we have the following trivial
lower bound:
\[ \sum_{Q/2 < q\leq Q} \sum_{a\in(\Z/q\Z)^*} E(x;q,a)^2 \geq
   \big[\alpha F(Q)\big]\cdot(\log x)^{3A/2}.
\]
On the other hand, the Barban--Davenport--Halberstam theorem ensures that
the sum on the left-hand side is $\ll xQ/\log Q$. As a result, we get:
\[
    \alpha 
\ll \frac{(\log x)^{-3A/2}}{F(Q)} \cdot \frac{xQ}{\log Q} 
\ll \frac{x(\log x)^{-3A/2}}{Q\log Q}
\ll \frac{x(\log x)^{-3A/2}}{x(\log x)^{-A+1}}
\ll \frac1{(\log x)^{A/2}}.
\]
This allows us to deduce the analogue of Theorem~\ref{th:term2erh} for
Algorithm~\ref{alg:uncond}.
\setcounter{section}{3}
\setcounter{subsection}{5}
\begin{thm}
\label{th:termuncond}
Algorithm~\ref{alg:uncond} with $A>6$ terminates almost surely, requires
$(1+o(1))\varphi(q)/q\cdot\log x$ iterations of the main loop on average,
and consumes:
\[ \big(A + o(1)\big)\cdot\frac{\varphi(q)}{q}\cdot\frac{\log x\log\log x}{\log 2}
\]
bits of randomness on average.
\end{thm}
\begin{proof}
Algorithm~\ref{alg:uncond} terminates almost surely because the trivial
algorithm does. One can estimate its average number of iterations as
follows. Denote by $\varpi(t)$ the probability that
Algorithm~\ref{alg:uncond} terminates after exactly $t$ iterations, and
$\varpi_{q,a}(t)$ the probability of the same event conditionally to the
pair $(q,a)$ being chosen in
Steps~\ref{algstep:uncond_aq}--\ref{algstep:uncond_aqcheck}. We have:
\begin{align*}
  \varpi_{q,a}(t) &= \begin{cases}
     \displaystyle
     \bigg(1 - \frac{\pi(x;q,a)}{1+\lfloor\frac{x-a}q\rfloor}\bigg)^{t-1}
     \cdot\frac{\pi(x;q,a)}{1+\lfloor\frac{x-a}q\rfloor}
     & \text{for $t \leq T$;} \\
     \displaystyle
     \bigg(1 - \frac{\pi(x;q,a)}{1+\lfloor\frac{x-a}q\rfloor}\bigg)^{T}
     \bigg(1 - \frac{\pi(x)}{\lfloor x\rfloor}\bigg)^{t-T-1}
     \cdot\frac{\pi(x)}{\lfloor x\rfloor}
     & \text{otherwise.}
   \end{cases} \\
  \varpi(t) &= \frac1{F(Q)}\sum_{Q/2< q\leq Q}\sum_{a\in(\Z/q\Z)^*}
\varpi_{q,a}(t).
\end{align*} 
Moreover, the expected number $N$ of iterations in
Algorithm~\ref{alg:erh} is given by $N = \sum_{t\geq 1} t\varpi(t)$.
We can denote by $N_{q,a} = \sum_{t\geq 1} t\varpi_{q,a}(t)$ the contribution of
a certain choice $(q,a)$.

As we have just seen, $\pi(x;q,a) = (1+o(1))\pi(x)/\varphi(q)$, except
perhaps for a fraction $\alpha\ll (\log x)^{-A/2}$ of all choices of
$(q,a)$. If we denote by $A$ the set of ``bad'' choices of $(q,a)$, we
can write, as before, that for all $(q,a)\in A$:
\[
  \varpi_{q,a}(t) \leq \begin{cases}
     \displaystyle
     1
     & \text{for $t \leq T$;} \\
     \displaystyle
     \Big(1 - \frac{\pi(x)}{\lfloor x\rfloor}\Big)^{t-T-1}
     \cdot\frac{\pi(x)}{\lfloor x\rfloor}
     & \text{otherwise.}
   \end{cases} \\
\]
Hence, if we let $\xi = \pi(x)/\lfloor x\rfloor$, we again obtain:
\[
 N_{q,a} \leq \sum_{t=1}^T t + \sum_{t=T+1}^{+\infty} t(1-\xi)^{t-T-1} \xi
 \leq T(T+3)/2 + \frac1\xi \ll \log^4 x.
\]
On the other hand, for $(q,a)\not\in A$, we have $\xi_{q,a} :=
\frac{\pi(x;q,a)}{1+\lfloor\frac{x-a}q\rfloor} =
\frac{q}{\varphi(q)}\cdot\frac{1+o(1)}{\log x}$. Therefore, as before:
\[
  \Big| N_{q,a} - \frac1{\xi_{q,a}} \Big| \leq
	 \exp\Big(-(1+o(1))\frac{q}{\varphi(q)}\log x\Big)
	 \cdot(2T+1/\xi+1/\xi_{q,a}) \ll \frac1{x^{1-\ve}}.
\]
As a result, we get:
\begin{align*}
  N &= \frac1{F(Q)}\sum_{Q/2<q\leq Q}\sum_{a\in(\Z/q\Z)^*} N_{q,a}\\
    &= \Big(1-O\big((\log x)^{-A/2}\big)\Big)\cdot\Big(\frac1{\xi_{q,a}} +
       O(1/x^{1-\ve})\Big) + O\big((\log x)^{-A/2}\big)\cdot O(\log^2 x) \\
    &= \frac1{\xi_{q,a}} + O\big((\log x)^{4-A/2}\big)
     = (1+o(1))\frac{\varphi(q)}q\cdot \log x
\end{align*}
as required, since $4-A/2 < 1$.
As for the expected number $R$ of random bits consumed by
the algorithm, it is now given by:
\[ R = \frac{A\log\log x}{\log 2} \sum_{t=1}^T t\cdot\varpi(t) +
 \sum_{t=T+1}^{+\infty} \Big(T\cdot\frac{A\log\log x}{\log 2} + (t -
T)\cdot\frac{\log x}{\log 2}\Big)\cdot\varpi(t)
\]
where we have again ignored the random bits necessary to pick the pair $(q,a)$,
since we need only $\frac{Q(3Q+2)}{8F(Q)} \sim \pi^2/6$ iterations
of the loop from Step~\ref{algstep:uncond_aq} to
Step~\ref{algstep:uncond_aqcheck} to select it, and hence $O(\log x)$
random bits. The stated estimate is obtained by
essentially the same computation as for $N$.\qed
\end{proof}

We now turn to estimates on the uniformity of the output distribution of
the algorithm. For that purpose, we consider instead the output
distribution $X$ of Algorithm~\ref{alg:uncondnofall}, the variant of
Algorithm~\ref{alg:uncond} in which no bound is set to the number of
iterations of the main loop
(Steps~\ref{algstep:uncond-loopstart}--\ref{algstep:uncond-loopend}),
{i.e.} with no fallback to the trivial algorithm. Clearly, since the
trivial algorithm has a perfectly uniform output distribution, the 
output distribution of Algorithm~\ref{alg:uncond} is at least as close to
uniform as $X$.

\begin{algorithm}[t]
\caption{An variant of Algorithm~\ref{alg:uncond} with no fallback.}
\label{alg:uncondnofall}
\begin{algorithmic}[1]
  \State Fix $Q \propto x(\log x)^{-A}$ even
  \State $q \randgets \{Q/2+1,\dots,Q\}$
  \State $a \randgets \{0,\dots,q-1\}$
  \State \textbf{if} {$\gcd(a,q)\neq1$}
    \textbf{then goto} step \ref{algstep:uncond_aq}
  \Repeatx{forever}
    \label{algstep:uncond-loopstart}
    \State $t \randgets \{0,\dots,\lfloor \frac{x-a}q \rfloor\}$
    \State $p \gets a + t\cdot q$
    \Ifx{$p$ is prime}
      \textbf{return} $p$
  \EndRepeatx
    \label{algstep:uncond-loopend}
\end{algorithmic}
\end{algorithm}

\begin{thm}
\label{th:statuncond}
The output distribution
of Algorithm~\ref{alg:uncondnofall} (and hence also
Algorithm~\ref{alg:uncond}) has a statistical distance $\Delta_1\ll
(\log x)^{(1-A)/2}$ to the uniform distribution. In particular, this
distance is polynomially small as long as $A>1$.
\end{thm}
\begin{proof}
Algorithm~\ref{alg:uncondnofall} produces an output (almost surely)
for exactly those choices of $(q,a)$ such that $\pi(x;q,a)\neq 0$. Let us
denote by $F^*_x(Q)$ the number of such choices. Clearly, with
notation from above, we have:
\[ F^*_x(Q) = \sum_{Q/2<q\leq Q} \varphi^*_x(q)\qquad\textrm{and}\qquad
   1-\alpha \leq \frac{F^*_x(Q)}{F(Q)} \leq 1.
\]
Then, the probability that Algorithm~\ref{alg:uncondnofall} outputs a
given prime $p\leq x$ can be written as:
\[ \Pr[X=p] = \frac1{F^*_x(Q)} \sum_{\substack{Q/2<q\leq Q\\ p\nmid q}}
     \frac1{\pi(x;q,p\bmod q)}.
\]
Therefore, we have:
\begin{align*}
   \Delta_1(X)
&= \sum_{p\leq x} \Big| \Pr[X=p] - \frac1{\pi(x)} \Big| \\
&= \sum_{p\leq x} \Bigg| \frac1{F^*_x(Q)} \sum_{\substack{Q/2<q\leq Q\\
   p\nmid q}} \frac1{\pi(x;q,p\bmod q)} - \frac1{\pi(x)} \Bigg| \\
&= \frac1{F^*_x(Q)} \sum_{p\leq x} \Bigg| \sum_{\substack{Q/2<q\leq Q\\
   p\nmid q}} \frac1{\pi(x;q,p\bmod q)} - 
   \sum_{Q/2<q\leq Q}\frac{\varphi^*_x(q)}{\pi(x)} \Bigg| \\
&\leq \frac1{F^*_x(Q)} \sum_{p\leq x} \sum_{\substack{Q/2<q\leq Q\\
   p\nmid q}} \Big| \frac1{\pi(x;q,p\bmod q)} -
   \frac{\varphi^*_x(q)}{\pi(x)} \Big| + \frac1{F^*_x(Q)} \sum_{p\leq x}
   \sum_{\substack{Q/2<q\leq Q\\ p|q}} \frac{\varphi^*_x(q)}{\pi(x)} \\
&\leq \frac1{F^*_x(Q)} \sum_{Q/2<q\leq Q} \sum_{\substack{p\leq x\\
   p\nmid q}} \Big| \frac1{\pi(x;q,p\bmod q)} -
   \frac{\varphi^*_x(q)}{\pi(x)} \Big| + \frac1{F^*_x(Q)} \sum_{Q/2<q\leq
   Q} \frac{\omega(q)\varphi^*_x(q)}{\pi(x)} \\
&\leq \frac1{F^*_x(Q)\pi(x)} \sum_{\substack{Q/2<q\leq Q\\
   a\in(\Z/q\Z)^*
   }} \Big| \pi(x) -
   \varphi^*_x(q)\pi(x;q,a) \Big| + \frac{O(\log Q)}{\pi(x)}.
\end{align*}
We can then bound the sum over $(q,a)$ of $\big| \pi(x) -
\varphi^*_x(q)\pi(x;q,a) \big|$ as $D+D^*$, where:
\[ D = \sum_{\substack{Q/2<q\leq Q\\ a\in(\Z/q\Z)^* }}
         \Big| \pi(x) - \varphi(q)\pi(x;q,a) \Big|
   \quad\text{and}\quad
   D^* = \sum_{\substack{Q/2<q\leq Q\\ a\in(\Z/q\Z)^* }}
         \big| \varphi(q) - \varphi^*_x(q) \big|\cdot\pi(x;q,a).
\]
Now, on the one hand:
\begin{align*}
   D^*
&= \sum_{Q/2<q\leq Q} \big| \varphi(q) - \varphi^*_x(q) \big|
   \sum_{a\in(\Z/q\Z)^*} \pi(x;q,a) \\
&\leq \big( F(Q) - F^*_x(Q) \big) \cdot \pi(x)
 \leq \alpha F(Q)\pi(x)
 \ll  x^3(\log x)^{-5A/2-1},
\end{align*}
and on the other hand, applying the Cauchy--Schwarz inequality and the
Barban--Davenport--Halberstam theorem:
\begin{align*}
   D
&= \sum_{\substack{Q/2<q\leq Q\\ a\in(\Z/q\Z)^* }}
     \varphi(q)\cdot\Big| \pi(x;q,a) - \frac{\pi(x)}{\varphi(q)} \Big| 
 \leq Q \sum_{\substack{Q/2<q\leq Q\\ a\in(\Z/q\Z)^* }}
          \Big| \pi(x;q,a) - \frac{\pi(x)}{\varphi(q)} \Big| \\
&\leq Q \cdot \sqrt{F(Q)} \cdot \sqrt{\frac{xQ}{\log Q}}
 \ll  x^3(\log x)^{-5A/2-1/2}.
\end{align*}
As a result, we obtain:
\[ \Delta_1(X) \ll \frac1{F^*_x(Q)\pi(x)}\cdot(D+D^*)
	\ll \frac{(\log x)^{2A+1}}{x^3}\cdot\frac{x^3}{(\log
	x)^{5A/2+1/2}} \ll \frac1{(\log x)^{(A-1)/2}}.
\]\qed

\end{proof}

\section{Comparison with other prime number generation algorithms}

In this section, we compare other prime number generation algorithms to
our method. In \S\ref{sec:why} we show that the output distribution of
Brandt and Damg\aa{}rd's PRIMEINC algorithm exhibits significant biases
(a fact that is intuitively clear, but which we make quite precise), and
in \S\ref{sec:maurer}, we discuss the advantages and drawbacks of our
method compared to that of Maurer, as the only previous work emphasizing
a close to uniform output distribution.

\subsection{PRIMEINC}
\label{sec:why}

Previous works on the generation of prime numbers, such as
\cite{DBLP:conf/crypto/BrandtD92,DBLP:conf/ches/JoyeP06},
provide a proof (based on rather strong assumptions) that the
output distribution of their algorithm has an entropy not much smaller
than the entropy of the uniform distribution. This is a reasonable
measure of the inability of an adversary to guess which particular prime
was output by the algorithm, but it doesn't rule out the possibility of
gaining some information about the generated primes. In particular, it
doesn't rule out the existence of an efficient distinguisher between the
output distribution and the uniform one.

Consider for example the PRIMEINC algorithm studied by Brandt and
Damg\aa{}rd in~\cite{DBLP:conf/crypto/BrandtD92}. In essence, it
consists in picking a random integer $y < x$ and returning the smallest
prime greater or equal to $y$. There are some slight technical
differences between this description and the actual PRIMEINC\footnote{To
wit, Brandt and Damg\aa{}rd restrict their attention to odd numbers
$x/2<y<x$, and set an upper bound to the number of iterations in the
algorithm}, but they have essentially no bearing on the following
discussion, so we can safely ignore them.

It is natural to suspect that the distribution of the output of this
algorithm is quite different from the uniform distribution: for example,
generating the second prime of a twin prime pair, i.e. a prime $p$ such
that $p-2$ is also prime, is abnormally unlikely. More precisely, if one
believes the twin prime conjecture, the proportion of primes of that form
among all primes up to $x$ should be $\sim 2c_2/\log x$, where
$c_2\approx 0.66$ is the twin prime constant. On the other hand, PRIMEINC
outputs such a $p$ if and only if it initially picks $y$ as $p$ or $p-1$.
Therefore, we expect the frequency of such primes $p$ in the output of
PRIMEINC to be much smaller, about $4c_2/(\log x)^2$. This provides an
efficient distinguisher between the output of PRIMEINC and the uniform
distribution.

More generally, it is easy to see that the method used
in~\cite{DBLP:conf/crypto/BrandtD92} to obtain the lower bound on the
entropy of the output distribution of PRIMEINC can also provide a
relatively large constant lower bound on the statistical distance to the
uniform distribution. Indeed, the method relies on the following result,
obtained by a technique first proposed by Gallagher~\cite{Gallagher}.

\def\mycite{\cite[Lemma 5]{DBLP:conf/crypto/BrandtD92}}
\begin{lemma}[\mycite]
\label{lem:gallagher}
Assume the prime $r$-tuple conjecture, and let $F_h(x)$ denote the number
of primes $p\leq x$ such that the largest prime $q$ less than $p$ satisfies
$p-q\leq h$. Then for any constant $\lambda$,
\[ F_{\lambda\log x}(x) = \frac{x}{\log
x}\big(1-e^{-\lambda}\big)(1+o(1)) \]
as $x\to+\infty$.
\end{lemma}

Now, the probability that the PRIMEINC algorithm outputs a fixed $p$ can
clearly be written as $d(p)/x$, where $d(p)$ is the distance between $p$
and the prime that immediately precedes it. Let $\Delta_1'$ be the
statistical distance between the output distribution of PRIMEINC and the
uniform distribution. We have:
\[ \Delta_1' = \sum_{p\leq x} \left|\frac{d(p)}{x} - \frac{1}{\pi(x)}\right|
     >\!\!\! \sum_{\substack{p\leq x \\ d(p)>2\log x}}
             \left|\frac{2\log x}{x} - \frac{\log x}{x}\right|
\]
for $x\geq 17$ in view of the already mentioned classical bound
$\pi(x)>x/\log x$ for $x\geq 17$. By Lemma~\ref{lem:gallagher}, this
gives:
\[ \Delta_1' > \frac{\log x}{x} F_{2\log x}(x) =
\big(1-e^{-2}\big)(1+o(1)) > 0.86 + o(1)
\]
as $x\to+\infty$, and in particular, $\Delta_1'$ admits a relatively large
constant lower bound (at least if one believes the prime $r$-tuple
conjecture on which the entropy bound is based).

Since the entropy lower bounds for other prime generation algorithms
like~\cite{DBLP:conf/ches/JoyeP06} are based on the same techniques, it
is very likely that their output distributions can similarly be shown to
be quite far from uniform.

\subsection{Maurer's algorithm}
\label{sec:maurer}

In~\cite{DBLP:conf/eurocrypt/Maurer89,DBLP:journals/joc/Maurer95}, Maurer
proposed a prime generation algorithm based on a similar principle as
Bach's technique to produce random numbers with known
factorization~\cite{DBLP:journals/siamcomp/Bach88}. This algorithm has
the major advantage of producing a primality certificate as part of its
output; as a result, it generates only provable primes, contrary to our
method.

Additionally, in the standard version of the algorithm, the distribution
of generated primes is heuristically close to uniform. More precisely,
this distribution would be exactly uniform if the following assertions
held:
\begin{enumerate}
\item \label{item:twoxplusone}
  The distribution of the relative sizes of the prime factors of
  an integer $x$ of given length \emph{conditional to} $2x+1$ being prime
  is the same as the distribution of the relative sizes of the prime
  factors of a uniformly random integer of the same length.

\item \label{item:dickson}
  That distribution is, moreover, the same as the asymptotic one (i.e. as
  the length goes to $+\infty$), as computed using Dickson's $\rho$
  function.
\end{enumerate}
Assertion~\ref{item:twoxplusone} is a rather nonstandard statement about
prime numbers, but Maurer points to some heuristic arguments that renders
it quite plausible \cite{Maurer92e}, at least asymptotically.
Assertion~\ref{item:dickson} is of course not exactly true, and it seems
difficult to quantify how much this affects the output distribution, but
probably not to a considerable extent.

That standard version of the algorithm, however, has a small probability
of not terminating, as discussed in
\cite[\S3.3]{DBLP:journals/joc/Maurer95}. To circumvent that problem,
Maurer suggests a modification to the algorithm which introduces biases
in the output distribution. If it is used, the algorithm will in
particular never generate primes $p$ such that $(p-1)/2$ has a large
prime factor. This provides an efficient distinguisher from the uniform
distribution (exactly analogous to the ``twin prime'' distinguisher of
the previous section), since e.g. safe primes have non negligible
density. There are other ways to avoid cases of non-termination that do
not affect the output distribution to such an extent (such as some form
of backtracking), but Maurer advises against them because of the large
performance penalty they may incur.

Furthermore, the paper \cite{DBLP:journals/joc/Maurer95} also describes
faster variants of the algorithm that are less concerned with the quality
of the output distribution, and they do indeed have outputs that are very
far from uniform, typically reaching only around 10\% of all prime
numbers.

More importantly, the main drawback of Maurer's algorithm compared to our
method is likely the amount of randomness it consumes: it is within a
small constant factor of the amount consumed by the trivial algorithm. In
fact, if we neglect everything but the last loop of the topmost step in
the recursion, and count as zero the cases when the recursion has more
than two steps, we see that the average amount of randomness used by the
algorithm is bounded below by $c$ times the corresponding amount for the
trivial algorithm, where (with the notations of
\cite[Appendix~1]{DBLP:journals/joc/Maurer95}):
\[ c \geq \int_{1/2}^1 \big(1-x\big)\, dF_1(x)
        = \int_{1/2}^1 \big(1-x\big)\, \frac{dx}{x}
        = \log 2 - \frac12
     \geq 0.19
\]
and we only counted a very small fraction of the actual randomness used!

As a result, in context where random bits are scarce and primality
certificates are not required, it seems that our method offers a rather
better trade-off than Maurer's algorithm: the uniformity of the
distribution is easier to estimate in our case, and we use much less
randomness. Moreover, Maurer finds
\cite[\S4.3]{DBLP:journals/joc/Maurer95} that an efficient implementation
of his algorithm is about 40\% slower than the trivial algorithm, whereas
Algorithm~\ref{alg:basic} is easily 5 times as fast in practice.

On the other hand, if proofs of primality are considered important, it is
clearly better to rely on Maurer's algorithm than to use our method and
then primality proving.

\bibliographystyle{abbrv}
\bibliography{biblio}

\end{document}